\documentclass[runningheads,11points]{article}

\usepackage[english]{babel}
\usepackage[utf8x]{inputenc}

\usepackage{amssymb,amstext,amsmath,amsfonts}
\usepackage{cite,scrtime}
\usepackage{graphics,graphicx}
\usepackage{dsfont}
\usepackage{mathrsfs}
\usepackage{fancyhdr}
\usepackage{verbatim}
\usepackage{boxedminipage}
\usepackage{colortbl}
\usepackage{algorithm}
\usepackage[noend]{algorithmic}
\usepackage{mathrsfs}
\usepackage[breaklinks, bookmarks=false]{hyperref}
\usepackage{stmaryrd} 
\usepackage[T1]{fontenc}
\usepackage[affil-sl]{authblk} 
\usepackage{xspace}

\newcommand{\Acal}{\mathcal{A}}

\newcommand{\Ocal}{\mathcal{O}}

\newcommand{\Zbb}{\mathbb{Z}}

\definecolor{gray1}{gray}{0.775}

\newcommand{\probl}[3]{
\begin{flushleft}
\fbox{
\begin{minipage}{13cm}
\noindent {\sc #1}\\
          {\bf Input:} #2\\
          {\bf Output:} #3
\end{minipage}}
\medskip
\end{flushleft}
}

\newcommand{\tw}{{{\sf tw}}}

\newcommand{\es}{\emptyset}

\definecolor{pink}{rgb}{0.858, 0.188, 0.478}

\newcommand{\fpt}{{\sf FPT}\xspace}
\newcommand{\W}{{\sf W}\xspace}
\newcommand{\np}{{\sf NP}\xspace}

\newcommand{\xp}{{\sf XP}\xspace}

\usepackage{aecompl}
\usepackage{color}
\definecolor{linkcol}{rgb}{0,0,0.8}
\definecolor{citecol}{rgb}{0.65,0,0}
\definecolor{titlecol}{rgb}{0.65,0,0}

\hypersetup {
pdfmenubar=true, 
pdfhighlight=/O, 
colorlinks=true, 
pdfpagemode=None, 
pdfpagelayout=DoublePage, 
pdffitwindow=true, 
linkcolor=linkcol, 
citecolor=citecol, 
urlcolor=linkcol 
}

\newcommand{\TRUE}[1]{\textsc{True}}
\newcommand{\FALSE}[1]{\textsc{False}}

\newtheorem{theorem}{Theorem}
\newtheorem{proposition}{Proposition}
\newtheorem{lemma}{Lemma}

\newcommand{\dist}[2][]{\ensuremath{\operatorname{dist}}\ifx\relax#1\relax\else\ensuremath{_{#1}}\fi\ensuremath{(#2)}}
\newcommand{\degree}[2][]{\ensuremath{\operatorname{deg}}\ifx\relax#1\relax\else\ensuremath{_{#1}}\fi\ensuremath{(#2)}}

\newcommand{\pdist}[2][]{\ensuremath{\operatorname{pdist}}\ifx\relax#1\relax\else\ensuremath{_{#1}}\fi\ensuremath{(#2)}}

\newenvironment{proof}[1][]{\par \noindent {\bf Proof:#1}\ }{\hfill$\Box$\\}

\DeclareMathOperator{\po}{\overrightarrow{\chi}} 

\title{Weighted proper orientations of trees\\ and graphs of bounded treewidth\thanks{Research supported by French grants DE-MO-GRAPH ANR-16-CE40-0028 and ESIGMA ANR-17-CE40-0028, and by CNPq-Brazil Universal projects 459466/2014-3,  310234/2015-8, and 401519/2016-3.}
}

\author[1]{J\'ulio Ara\'ujo}
\author[2]{Cl\'audia Linhares Sales}
\author[1,3]{Ignasi Sau}
\author[1]{Ana Silva}

\affil[1]{\emph{\small Departamento de Matem\'atica, Universidade Federal do Cear\'a, Fortaleza, Brazil}}
\affil[2]{\emph{\small Departamento de Computa\c{c}\~{a}o, Universidade Federal do Cear\'a, Fortaleza, Brazil}}
\affil[3]{\emph{\small CNRS, LIRMM, Universit\'e de Montpellier, Montpellier, France\newline \texttt{julio@mat.ufc.br, linhares@lia.ufc.br, ignasi.sau@lirmm.fr, anasilva@mat.ufc.br}}}

\begin{document}

\maketitle
\setcounter{footnote}{0}

\vspace{-.25cm}


\begin{abstract}
Given a simple graph $G$, a weight function $w:E(G)\rightarrow \mathbb{N} \setminus \{0\}$, and an orientation $D$ of $G$, we define $\mu^-(D) = \max_{v \in V(G)} w_D^-(v)$, where $w^-_D(v) =  \sum_{u\in N_D^{-}(v)}w(uv)$. We say that $D$ is a \emph{weighted proper orientation} of $G$ if $w^-_D(u) \neq w^-_D(v)$ whenever $u$ and $v$ are adjacent.   We introduce the parameter  {\em weighted proper orientation number} of $G$, denoted by $\po(G,w)$, which is the minimum, over all weighted proper orientations $D$ of $G$, of $\mu^-(D)$. When all the weights are equal to 1, this parameter  is equal to the {\em proper orientation number} of $G$, which has been object of recent studies and whose determination is \np-hard in general, but polynomial-time solvable on trees.  Here, we prove that the equivalent decision problem of the weighted proper orientation number (i.e., $\po(G,w) \leq k?$) is (weakly) \np-complete on trees but can be solved by a pseudo-polynomial time algorithm whose running time depends on $k$. Furthermore, we present a dynamic programming algorithm to determine whether a general graph $G$ on $n$ vertices and treewidth at most $\tw$ satisfies $\po(G,w) \leq k$, running in time $\Ocal(2^{\tw^2}\cdot k^{3\tw}\cdot \tw \cdot n)$, and we complement this result by showing that the problem is $\W[1]$-hard on general graphs parameterized by the treewidth of $G$, even if the weights are polynomial in $n$.

\vspace{0.25cm}


\end{abstract}

\section{Introduction}
\label{sec:intro}

Let $G=(V,E)$ be a simple graph. We refer the reader to~\cite{BM08} for the usual definitions and terminology in graph theory.  In this paper, we denote by $(G,w)$ an edge-weighted graph $G$, where $w:E(G)\rightarrow \mathbb{N} \setminus \{0\}$. For an edge $e=uv$ of $G$, we write $w(e)$ or $w(uv)$, indistinctly, to denote its weight. In a digraph $D$, the notation $\overrightarrow{uv}$ means an arc with  tail $u$ and head $v$. An \emph{orientation} $D$ of $G$ is a digraph obtained from $G$ by replacing each edge $uv$ of $G$ by exactly one of the arcs  $\overrightarrow{uv}$ or $\overrightarrow{vu}$. For a vertex $v$, $N_D^{-}(v)$ (resp. $N_D^{+}(v)$) is the set of the neighbors $w$ of $v$ such that $\overrightarrow{wv}$ (resp. $\overrightarrow{vw}$) is an arc of $D$. The \emph{indegree} (resp. \emph{outdegree})  of $v$, denoted by $d^-(v)$ (resp. $d^+(v)$), is the cardinality of $N_D^-(v)$ (resp. $N_D^{+}(v)$). Observe that $N(v) = N_D^-(v) \cup N_D^{+}(v)$ for any orientation $D$ of $G$.
The \emph{inweight} (resp. \emph{outweight}) of $v$, denoted by $w^-(v)$ (resp. $w^+(v)$), is the value $\sum_{u\in N_D^{-}(v)}w(uv)$ (resp. $\sum_{u\in N_D^{+}(v)}w(uv)$). Whenever it is clear from the context, the subscript $D$ will be omitted. We denote by $\mu^-(D)$ the maximum inweight of $D$.

A \emph{proper coloring of $G$} is a function $f:V(G)\rightarrow \mathbb{N}$ such that $f(u)\neq f(v)$ for every $uv\in E(G)$. Given an edge-weighted graph $(G,w)$, a \emph{weighted proper orientation of $G$} is an orientation $D$ of $G$ such that $w^-(u)\neq w^-(v)$ for every $uv\in E(G)$ (i.e., the inweights of the vertices define a proper coloring of $G$). We define $\po(G,w)$ as the minimum, over all weighted proper orientations $D$ of $G$, of $\mu^-(D)$.  Note that if $w(e)=1$ for every edge $e \in E(G)$, then the parameter $\po(G,w)$ is equal to the \emph{proper orientation number} of $G$, denoted by $\po(G)$ and recently studied in a series of articles~\cite{AD13, ACR+15, AHLS16,KMM+17}.

This latter parameter was introduced by Ahadi and Dehghan~\cite{AD13} in 2013.
They observed that this parameter is well-defined for any graph $G$, since one can always obtain a proper orientation $D$ with $\mu^{-}(D) \leq \Delta(G)$, where $\Delta(G)$ is the maximum degree of $G$. They also proved that deciding whether a graph $G$ has proper orientation number equal to 2 is \np-complete even if $G$ is a planar graph. Other complexity results were obtained by Ara\'ujo et al.~\cite{ACR+15}. They proved that the problem of determining the proper orientation number of a graph remains \np-hard for subclasses of planar graphs that are also bipartite and of bounded degree. In the same paper, they proved that the proper orientation number of any tree is at most 4; Knox et al.~\cite{KMM+17} provided a shorter proof of the same result. In another paper, Ara\'ujo et al.~\cite{AHLS16} proved that the proper orientation number of cacti is at most 7, and that this bound is tight.

All the above negative results also apply to the \textsc{Weighted Proper Orientation} number problem, whose corresponding decision problem is formally defined as follows:

\probl{Weighted Proper Orientation}{An edge-weighted graph $(G,w)$ and a positive integer $k$.}{Is $\po(G,w)\leq k$?}

Let us now discuss our motivation to introduce the weighted version of the proper orientation number. It is claimed in~\cite{ACR+15}, without a proof, that ``one can observe that, for fixed integers $t$ and $k$, determining whether $\po(G) \leq k$ in a graph $G$ of treewidth at most $t$ can be done in polynomial time using a standard dynamic programming approach''. This implies, in particular, that one can determine the proper orientation number of a tree in polynomial time.

Even if one may think that most problems are easily solvable on trees, there are scarce but relevant counterexamples: Ara\'ujo et al.~\cite{AraujoNP14} showed that the \textsc{Weighted Coloring} problem on $n$-vertex trees cannot be solved in time $2^{o(\log^2 n)}$ unless the Exponential Time Hypothesis of Impagliazzo et al.~\cite{ImpagliazzoP01} fails. It is remarkable that this bound is tight, in the sense that the problem can be solved in time $2^{\Ocal(\log^2 n)}$. Further hardness results for \textsc{Weighted Coloring} on trees and forests under the viewpoint of parameterized complexity were recently given by Ara\'ujo et al.~\cite{AraujoBS17}.

It turns out that \textsc{Weighted Proper Orientation} constitutes another example of a coloring problem that is hard on trees: we prove that the problem is \np-complete on trees, by a reduction from the \textsc{Subset Sum} problem. Since \textsc{Subset Sum} is a well-known example of {\sl weakly} \np-complete problem that can be solved in pseudo-polynomial time~\cite{GaJo79}, a natural question is whether the \textsc{Weighted Proper Orientation} problem on trees exhibits the same behavior. Our main technical contribution is a positive answer to this question. Interestingly, our pseudo-polynomial algorithm uses as a black box a subroutine to solve an appropriately defined \textsc{Subset Sum} instance. Another ingredient of this algorithm is a combinatorial lemma stating that there always exists a weighted proper orientation $D$ of a tree $T$ such that $d^-_D(u)\le 4$, for every $u\in V(T)$; this generalizes the result of Ara\'ujo et al.~\cite{ACR+15} and Knox et al.~\cite{KMM+17} for the unweighted version.

%
%
%


After focusing on trees, we explore the complexity of \textsc{Weighted Proper Orientation} on the more general class of graphs of bounded treewidth. We first present a dynamic programming algorithm to determine whether an $n$-vertex edge-weighted graph $(G,w)$ with treewidth at most $\tw$  satisfies  $\po(G,w) \leq k$, running in time $\Ocal(2^{\tw^2}\cdot k^{3\tw}\cdot \tw \cdot n)$. In particular, when all weights are equal to 1, this algorithm finds in polynomial time the proper orientation number of graphs of bounded treewidth; as mentioned before, such algorithm had been claimed to exist in~\cite{ACR+15}.

Back to the weighted version, from the viewpoint of parameterized complexity~\cite{DF13,CyganFKLMPPS15}, the running time of our algorithm shows that the problem is in $\xp$ parameterized by the treewidth of the input graph. Hence, the natural question is whether it is  $\fpt$. We answer this question in the negative
by showing that the term $k^{\Ocal(\tw)}$ is essentially unavoidable, in the sense that, under the assumption that $\fpt \neq \W[1]$, there is no algorithm  running in time $f(\tw) \cdot (k \cdot n)^{\Ocal(1)}$ for any computable function $f$. We prove this via a parameterized reduction from the \textsc{Minimum Maximum Indegree} problem, known to be $\W[1]$-hard parameterized by the treewidth of the input graph~\cite{Szeider11}, and which is defined as follows\footnote{The original problem is defined in terms of outweight instead of inweight, but for convenience we consider the latter version here, which is clearly equivalent to the original one.}:

\probl{Minimum Maximum Indegree}
{An edge-weighted graph $(G,w)$ and a positive integer $k$.}
{Is there an orientation $D$ of $G$ such that $\mu^{-}(D) \leq k$?}

The above problem has been recently studied in the literature~\cite{Szeider11,AsahiroMO11,SzeiderACM11}, and its similarity to \textsc{Weighted Proper Orientation} can be considered as a further motivation to study the latter problem. It is worth pointing out that in our $\W[1]$-hardness reduction the edge-weights are polynomial in the size of the graph, implying that the pseudo-polynomial algorithm that we presented for trees cannot be generalized to arbitrary values of treewidth.

%
%

\medskip

The remainder of the article is organized as follows. In Section~\ref{sec:prelim} we recall the definition of (nice) tree decompositions and we present some basic preliminaries of parameterized complexity. In Section~\ref{sec:trees} we focus on trees and in Section~\ref{sec:bounded-tw} we turn our attention to graphs of bounded treewidth. We conclude this paper in Section~\ref{sec:further} with some open questions.



\section{Preliminaries}
\label{sec:prelim}

In this section, we recall definitions and introduce the terminology adopted on tree decomposition and parameterized complexity.

\vspace{0.5cm}

\noindent \textbf{Tree decompositions and treewidth.} Given a graph $G$, a \emph{tree decomposition of $G$}~\cite{RS+84} is a pair ${\cal T} = (T,(X_t)_{t\in V(T)})$, where $T$ is a rooted tree and $X_t$ is a subset of vertices of $G$, for every $t\in V(T)$, that satisfies the following conditions:
\begin{enumerate}
  \item $\bigcup_{t\in V(T)}X_t = V(G)$;
  \item There exists $t\in V(T)$ such that $\{u,v\}\subseteq X_t$, for every edge $uv\in E(G)$; and
  \item Let $t,t'\in V(T)$ and $t''\in V(T)$ be a vertex in the $(t,t')$-path in $T$. Then, $X_t\cap X_{t'}\subseteq X_{t''}$.
\end{enumerate}

We refer to the vertices of a tree decomposition as \emph{nodes}. A tree decomposition is called \emph{nice} if each non-leaf node $t\in V(T)$ can be classified into one of the following types:
\begin{enumerate}
\item \emph{Introduce node}: if $t$ has exactly one child $t'$ in $T$ and $X_t = X_{t'}\cup \{u\}$ for some $u\in V(G)$;
\item \emph{Forget node}: if $t$ has exactly one child $t'$ in $T$ and $X_t = X_{t'}\setminus \{u\}$ for some $u\in V(G)$; and
\item \emph{Join node}: if $t$ has exactly two children $t',t''$ in $T$, and $X_t = X_{t'} = X_{t''}$.
\end{enumerate}

The \emph{width of ${\cal T}$} is equal to the maximum size of a subset $X_t$ minus one, and the \emph{treewidth of $G$} is the minimum width of a tree decomposition of $G$. It is known that if $G$ has a tree decomposition of width $k$, then $G$ has a nice tree decomposition of width $k$ such that the tree has $\Ocal(k \cdot |V(G)|)$ nodes~\cite{Kloks94}.   Additionally, for a node $t \in V(T)$ we denote by $T_t$ the subtree of $T$ rooted at $t$, and by $G_t$ the subgraph of $G$ induced by $\bigcup_{t'\in V(T_t)}X_{t'}$.

\vspace{0.5cm}

\noindent \textbf{Parameterized complexity.} We refer the reader to~\cite{DF13,CyganFKLMPPS15} for basic background on parameterized complexity, and we recall here only some basic definitions.
A \emph{parameterized problem} is a language $L \subseteq \Sigma^* \times \mathbb{N}$.  For an instance $I=(x,k) \in \Sigma^* \times \mathbb{N}$, the value $k$ is called the \emph{parameter}; note that $x$ can be thought of as the instance of the associated unparameterized problem. 
A parameterized problem $L$ is \emph{fixed-parameter tractable} ({\sf FPT}) if there exists an algorithm $\Acal$, a computable function $f$, and a constant $c$ such that given an instance $I=(x,k)$ of $L$,
$\Acal$ (called an {\sf FPT} \emph{algorithm}) correctly decides whether $I \in L$ in time bounded by $f(k) \cdot |I|^c$.


%

Within parameterized problems, the class {\sf W}[1] may be seen as the parameterized equivalent to the class \np of classical optimization problems. Without entering into details (see~\cite{DF13,CyganFKLMPPS15} for the formal definitions), a parameterized problem being {\sf W}[1]-\emph{hard} can be seen as a strong evidence that this problem is {\sl not} \fpt. The canonical example of {\sf W}[1]-hard problem is \textsc{Independent Set} parameterized by the size of the solution\footnote{Given a graph $G$ and a parameter $k$, the problem is to decide whether there exists $S \subseteq V(G)$ such that $|S| \geq k$ and $E(G[S]) = \es$.}.


To transfer ${\sf W}[1]$-hardness from one problem to another, one uses a \emph{parameterized reduction}, which given an input $I=(x,k)$ of the source problem, computes in time $f(k) \cdot |I|^c$, for some computable function $f$ and a constant $c$, an equivalent instance $I'=(x',k')$ of the target problem, such that $k'$ is bounded by a function depending only on $k$. Hence, an equivalent definition of $\W$[1]-hard 
problem is any problem that admits a parameterized reduction from \textsc{Independent Set}
parameterized by the size of the solution.

Even if a parameterized problem is $\W$[1]-hard, it may still be solvable in polynomial time for \emph{fixed} values of the parameter: such problems are said to belong to the complexity class \xp. Formally,
a parameterized problem whose instances consist of a pair $(x,k)$ is in \xp if it can be solved by  an algorithm with
running time $f(k) \cdot |x|^{g(k)}$, where~$f,g$ are computable functions depending only on
the parameter and $|x|$ represents the input size. For example, \textsc{Independent Set}  parameterized by the solution size is easily seen to belong to \xp, as it suffices to check all the possible subsets of size $k$ of $V(G)$.

\section{Trees}
\label{sec:trees}

In this section we investigate the weighted proper orientation number of trees. We first generalize in Subsection~\ref{subsec:upper-trees} to the weighted version the fact that the proper orientation number of any tree is at most 4, which had been proved by Ara\'ujo et al.~\cite{AHLS16} and by Knox et al.~\cite{KMM+17}. This result will be then used in Subsection~\ref{sec:pseudo} to obtain a pseudo-polynomial algorithm on trees. Finally, we prove in Subsection~\ref{sec:hard-trees} that the problem is (weakly) \np-complete on trees, by a reduction from the \textsc{Subset Sum} problem.


\subsection{Upper bounds}
\label{subsec:upper-trees}

As mentioned before, the following is a generalization of previous results by Ara\'ujo et al.~\cite{ACR+15} and Knox et al.~\cite{KMM+17}. It will be used in the proof of Theorem~\ref{thm:pseudo-trees}.


\begin{lemma}\label{lem:boundedby4K}
Let $T$ be an edge-weighted tree and $w:E(T)\rightarrow \mathbb{N} \setminus \{0\}$. There exists a weighted proper orientation $D$ of $T$ such that $d^-_D(u)\le 4$, for every $u\in V(T)$.
\end{lemma}
\begin{proof}
If $T$ is a path, then there is nothing to prove, so suppose otherwise. Suppose $T$ is rooted at some vertex, and for each $v\in V(T)$ denote by $T_v$ the subtree rooted at $v$. We use induction on $|V(T)|$. For this, choose a vertex $v$ of degree at least~3 that it is closest to the leaves, i.e., at most one component of $T-v$ is not a path, namely the component containing the parent of $v$. Let $u$ be the parent of $v$, and $v_1,\ldots,v_q$ be its remaining neighbors.  For each $i\in \{1,\ldots,q\}$, denote by  $(A_i,B_i)$ the bipartition of the component of $T-v$ containing $v_i$, and suppose that $v_i\in A_i$. Also, denote by $w_i$ the value $w(vv_i)$ and suppose, without loss of generality, that $w_1\ge w_2\ge \ldots \ge w_q$. Now, let $D$ be a weighted proper orientation of $T-T_v$. We want to extend $D$ to a weighted proper orientation of $T$. First, orient $uv$ toward $v$, and let $w=w(uv)$ and $c=w+w_1+w_2$. We analyze the cases:
\begin{enumerate}
  \item $w^-(u)\neq c$: orient $vv_1$ and $vv_2$ toward $v$, and all the edges of $(A_i,B_i)$ from $A_i$ to $B_i$, for $i=1$ and $i=2$. Note that every vertex in $A_i$ have weight~0, and every vertex in $B_i$ have weight greater than~0. Thus, if $q=2$ we are done. Otherwise, consider any $i\in \{3,\ldots,q\}$. Orient $vv_i$ toward $v_i$. If $B_i=\emptyset$ (i.e., $v_i$ is a leaf), we are done since $w^-(v)=c > w_i = w^-(v_i)$. Otherwise, let $x_i$ be the neighbor of $v_i$ in $B_i$. If $w(v_ix_i)\neq c-w_i$, then orient the edges of $(A_i,B_i)$ from $B_i$ to $A_i$; all the vertices in $B_i$ have weight~0, all the vertices of $A_i$ have weight greater than~0, and $w^-(v_i) = w_i+w(v_ix_i)\neq c$. Otherwise, orient the edges of $(A_i,B_i)$ from $A_i$ to $B_i$. Similarly, the only possible conflict is between $v_i$ and $x_i$, which does not occur since $w^-(x_i) \ge c - w_i = w+w_1+w_2-w_i > w_i = w^-(v_i)$ (recall that $w_1\ge w_2\ge w_i$ and that $w\ge 1$).

 \item $w^-(u) = c$ and $q\ge 3$: in this case, the same argument as above can be applied, except that we orient $vv_1$, $vv_2$ and $vv_3$ toward $v$;

 \item $w^-(u) = c$ and $q=2$: for $i\in\{1,2\}$, if $B_i\neq \emptyset$, let $x_i$ be the neighbor of $v_i$ in $B_i$ and let $c_i=w(vv_i)+w(v_ix_i)$; otherwise, let $c_i = w(vv_i)$. If $c_i\neq w(uv)$ for $i=1$ and $i=2$, then for $i=1$ and $i=2$, orient $vv_i$ toward $v_i$, and orient the edges of $(A_i,B_i)$ from $B_i$ to $A_i$, if they exist. Now, let $i\in \{1,2\}$ be such that $c_i=w(uv)$, and let $j = 3-i$. Orient $vv_j$ toward $v$, $vv_i$ toward $v_i$, the edges of $(A_i,B_i)$ from $B_i$ to $A_i$, if any, and the edges of $(A_j,B_j)$ from $A_j$ to $B_j$, if any. The only possible conflict is between $v$ and $v_i$, which does not occur because $w^-(v) = w(uv)+w(vv_j)=c_i+w(vv_j)>c_i = w^-(v_i)$.
\end{enumerate}
\end{proof}


Note that the above lemma does not guarantee the existence of an {\sl optimal} weighted proper orientation $D$ such that $d^-_D(u)\le 4$, for every $u\in V(T)$; for instance, the tree constructed in the proof of Theorem~\ref{thm:weaklyNPc} does not admit, in general, such an optimal orientation.

\subsection{Pseudo-polynomial algorithm}
\label{sec:pseudo}

The next theorem proves the existence of an algorithm to solve the  \textsc{Weighted Proper Orientation} problem restricted to trees that runs in  polynomial time on the size of the input and the maximum weight. The algorithm crucially uses a subroutine to solve  the \textsc{Subset Sum} problem, which we recall here for completeness.

\probl{Subset Sum}{A set $S\subseteq \Zbb$ and a positive integer $k$.}{Does there exist $S'\subseteq S$ such that $\sum_{s\in S'} s = k$?}


We will use the shortcut $\textsc{SubsetSum}(S,\ell)$ to denote the \textsc{Subset Sum} problem with instance $(S, \ell)$. It is well-known~\cite{GaJo79} that this problem can be solved in pseudo-polynomial time $\Ocal(\ell \cdot |S|)$.

Before proving the theorem, we need an additional notation. Consider a rooted tree $T$ with root $r\in V(T)$, and assume that $r$ has  degree larger than~1.  For each $v\in V(T)$, we denote by $\pi(v)$ the parent of $v$ (consider $\pi(r)$ to be ${\sf null}$), by $T_v$ be the subtree rooted at $v$, and by $N_{T_v}(v)$ the neighbors of $v$ in $T_v$.


\begin{theorem}\label{thm:pseudo-trees}
Let $T$ be an edge-weighted tree on $n$ vertices,  $w:E(T)\rightarrow \mathbb{N} \setminus \{0\}$, and $K=\max_{e\in E(T)}w(e)$. It is possible to compute $\po(T,w)$ in time $\Ocal(K^3n)$.
\end{theorem}
\begin{proof}
Suppose that $|V(T)|\ge 3$, otherwise the problem is trivial. Consider $T$ to be rooted at $r\in V(T)$ with  degree larger than~1. We will prove that, given a positive integer $k\in\{K,\ldots,4K\}$, one can decide in time $\Ocal(k^2n)$ whether $\po(T,w)\le k$. Note that for smaller values of $k$, the answer is trivially ``{\sf no}'' and for larger values of $k$, it is ``{\sf yes}'' by Lemma~\ref{lem:boundedby4K}. Therefore, this fact will prove the theorem. The general idea is to construct a proper orientation for $T_v$ given appropriate proper orientations of $T_{v'}$ for each $v'\in N_{T_v}(v)$.

Given an orientation $D$ of $T_v$, where $v$ is any vertex of $T$, and a positive integer $\ell$, we say that $D+\ell$ is \emph{proper} if the function obtained from $w^-_D$ by adding $\ell$ to $w^-_D(v)$ is also a proper coloring of $T_v$, i.e., if $w^-_D(x)\neq w^-_D(y)$ for every $xy\in E(T_v)$, $x,y\neq v$, and $w^-_D(v)+\ell\neq w^-_D(x)$ for every $x\in N_{T_v}(v)$.
For each $v\in V(T)$ and each value $w\in\{0,\ldots,k\}$, we define the following parameters, which intuitively correspond to the existence of an appropriate orientation of $T_v$ where the edge $v\pi(v)$ is oriented away from $v$ or toward $v$, respectively:

\[\rho(v,w) = \left\{\begin{array}{ll}
1, & \mbox{if there exists a proper orientation $D$ of $T_v$ such that}\\
   & \mbox{ $w^-_D(v)=w$, and $\mu^-(D)\le k$;}\\
0, & \mbox{otherwise.}
\end{array}\right.\]

\[\rho'(v,w) = \left\{\begin{array}{ll}
1, &  \mbox{if there exists an orientation $D$ of $T_v$ such that }\\
   &  \mbox{ $D+w(v\pi(v))$ is proper, $w^-_D(v)=w-w(v\pi(v))$, and }\\
   &  \max\{w,\mu^-(D)\}\le k;\\
0, & \mbox{otherwise.}
\end{array}\right.\]

In the case of $r$, we get that $\rho(r, w)$ equals $\rho'(r, w)$, i.e., the parameter $\rho(r,w)$ indicates whether $T$ admits a proper orientation $D$ such that $w^-_D(r) = w$ and $\mu^-(D)\le k$. The answer to the problem is ``{\sf yes}'' if and only if $\rho(r,w) = 1$ for some $w\in \{0,\ldots,k\}$.

We now proceed to compute the parameters $\rho(v,w)$ and $\rho'(v,w)$, inductively from the leaves to the root of $T$. First, observe that if $v$ is a leaf, then $V(T_v) = \{v\}$, in which case we have:
\[\rho(v,w) = \left\{\begin{array}{ll}
1, & \mbox{if $w=0$, and}\\
0, & \mbox{otherwise.}
\end{array}\right.\]

\[\rho'(v,w) = \left\{\begin{array}{ll}
1, & \mbox{if $w=w(v\pi(v))$, and}\\
0, & \mbox{otherwise.}
\end{array}\right.\]

Now, let $v\in V(T)$ be a non-leaf vertex with neighbors $\{\pi(v),v_1,\ldots,v_q\}$, and suppose that we have already computed $\rho(v_i,w)$ and $\rho'(v_i,w)$, for every $i\in \{1,\ldots,q\}$ and every $w\in \{0,\ldots,k\}$. For each $w\in \{0,\ldots,k\}$,  we will use the already computed values to compute $\rho(v,w)$ and $\rho'(v,w)$.

To this end, consider a proper orientation $D$  of $T_v$ such that $w^-_D(v) = w$ and $\mu^-(D)\le k$. Also, for each $i\in\{1,\ldots,q\}$, let $D_i$ be the orientation $D$ restricted to $T_{v_i}$. For each $v_i$ such that $\overrightarrow{v_iv}\in D$, observe that $D_i$ is a proper orientation of $T_{v_i}$ such that $w^-_{D_i}(v_i) = w^-_{D}(v_i)$. This means that $\rho(v_i,w^-_{D}(v_i))=1$. Similarly, for each $v_i$ such that $\overrightarrow{vv_i} \in D$, we get that $D_i+w(v_iv)$ is proper and is such that $w^-_{D_i}(v_i) = w^-_D(v_i)-w(v_iv)$. Hence, $\rho'(v_i,w^-_{D}(v_i))=1$. Since $w^-_D(v_i)\neq w$, we are interested in the entries $\rho(v_i,w'),\rho'(v_i,w')$ such that $w'\neq w$. We then define the following subsets of $N_{T_v}(v)$:

\[F^+ = \{v_i, i\in\{1,\ldots,q\}\mid \rho(v_i,w')=0\mbox{ for every }w'\neq w\}\]
\[F^- = \{v_i, i\in\{1,\ldots,q\}\mid \rho'(v_i,w')=0\mbox{ for every }w'\neq w\}.\]

The discussion in the previous paragraph implies that, if the desired orientation exists, then it must necessarily contain the arcs $\{\overrightarrow{vv_i}\mid v_i\in F^+\}\cup \{\overrightarrow{v_iv} \mid v_i\in F^-\}$. Thus, if $F^+\cap F^-\neq \emptyset$, we can safely answer ``{\sf no}'', and this is why we can henceforth assume otherwise.

Let $N = \{v_1,\ldots,v_q\}\setminus (F^+\cup F^-)$, that is, the vertices in $N_{T_v}(v)$ for which some choice needs to be made; see Figure~\ref{fig:pseudo-polynomial} for an illustration. In order to compute the values of $\rho(v,w)$ and $\rho'(v,w)$, the orientation of the set of edges $\{ v_iv \mid v_i \in N\}$ will be chosen according to the output of an appropriate instance of the \textsc{Subset Sum} problem, constructed as follows.

 \begin{figure}[htb]
    \centering
    \includegraphics[scale=1.]{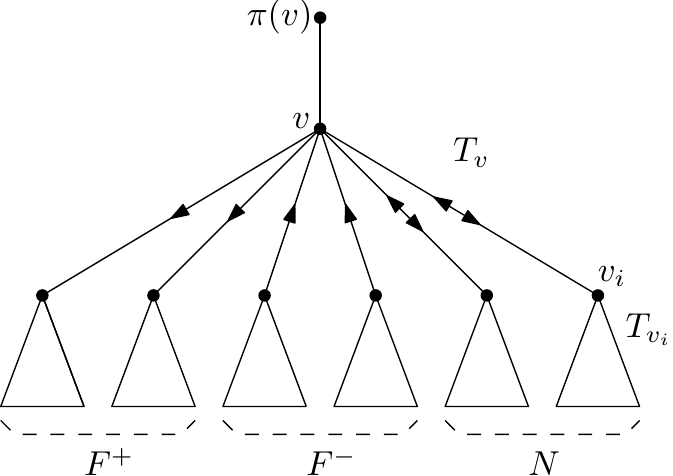}
    \caption{Illustration in the proof of Theorem~\ref{thm:pseudo-trees}.}
    \label{fig:pseudo-polynomial}
  \end{figure}

Note that if $D$ is the desired orientation, then it holds that $\sum_{v_i\in N^-_D(v)\setminus F^-}w(v_iv) = w-\sum_{v_i\in F^-}w(v_iv)=:\ell$. This means that the problem $\textsc{SubsetSum}(\{w(v_iv)\mid v_i\in N\},\ell)$ has a positive answer.
Conversely,  let $S$ be a subset that certifies a positive answer to $\textsc{SubsetSum}(\{w(v_iv)\mid v_i\in N\},\ell)$. Since $v_i\notin F^+$ for every $v_i\in S\cup F^-$, there exists a proper orientation $D_i$ of $T_{v_i}$ such that $\mu^-(D_i)\le k$ and $w^-_{D_i}(v_i)=w_i$ for some $w_i\neq w$. Also, since $v_i\notin F^-$ for every $v_i\in (N\setminus S)\cup F^+$, there exists an orientation $D_i$ of $T_{v_i}$ such that, for some $w_i\neq w$, we have $\max\{w_i,\mu^-(D_i)\}\le k$, $D_i+w(vv_i)$ is proper, and $w^-_{D_i}(v_i) = w_i-w(vv_i)$. One can verify that a proper orientation of $T_v$ can be obtained from $D_1,\ldots,D_q$ by orienting toward $v$ precisely the edges incident to $F^-\cup S$. 
Therefore, $\rho(v,w)=1$ if and only if $\textsc{SubsetSum}(\{w(v_iv) \mid v_i\in N\},\ell)$ has a positive answer.

On the other hand, by defining $F^+,F^-,N$ in the same way as above and letting $\ell'  := w(v\pi(v))+\sum_{v_i\in F^-}w(v_iv)$, using similar arguments one can prove that $\rho'(v,w)=1$ if and only if $\textsc{SubsetSum}(\{w(v_iv) \mid v_i\in N\},\ell')$ has a positive answer.

In both cases, since $\textsc{SubsetSum}(\{w(v_iv)\mid v_i\in N\},\ell)$ can be solved in time $\Ocal(\ell \cdot \lvert N\rvert) = \Ocal(k \cdot \deg_T(v))$, where $\deg_T(v)$ denotes the degree of $v$ in $T$, we get that computing the parameters $\rho(v,w)$ and $\rho'(v,w)$, for every $w\in\{0,\ldots,k\}$, takes time $\Ocal(k^2 \cdot \deg_T(v))$. This has to be done for every $v\in V(T)$ and every $k\in \{K,\ldots,4K\}$, yielding the claimed running time $\Ocal(K^3 \cdot \sum_{v \in V(T)}\deg_T(v)) = \Ocal(K^3 n)$.
\end{proof}

\subsection{\np-completeness}
\label{sec:hard-trees}

In this subsection, we reduce the \textsc{Subset Sum} problem to \textsc{Weighted Proper Orientation} on trees.
It is well-known that {\sc Subset Sum} is one of Karp's 21 {\np}-complete problems~\cite{Karp72} and that it is weakly \np-complete~\cite{GaJo79}.

\begin{theorem}\label{thm:weaklyNPc}
  {\sc Weighted Proper Orientation} is weakly \np-complete on trees.
\end{theorem}
\begin{proof}
  One can easily verify in linear time whether a given orientation $D$ is proper and whether $\mu^-(D) \leq k$. Thus, the problem is in \np.

  Let $S = \{i_1,\ldots, i_p\}\subseteq \Zbb$ and $k$ be an instance of the {\sc Subset Sum} problem, which is known to be $\np$-complete even if every $i_j\in S$ is a {\sl positive} integer~\cite{A98}. Hence, we assume that $k$ and all integers in $S$ are positive and that $i_j < k$, for every $j\in\{1,\ldots, p\}$. In the sequel, we construct a tree $T(S)$ and a function $w:E(T(S))\to \mathbb{N} \setminus \{0\}$ such that $(S,k)$ is a {\sf yes}-instance of {\sc Subset Sum} if, and only if, $((T(S),w), k')$ is a {\sf yes}-instance of  {\sc Weighted Proper Orientation}, where $k'= 2k+6$.

The tree $T(S)$ has $p$ vertices $v_j$, for every $j\in\{1, \ldots, p\}$, and $k+2$ paths on four vertices: a path $P^* = (w,w_1,w_2,w_3)$ and $k+1$ paths $P_\ell = (u_1^\ell,u_2^\ell,u_3^\ell,u_4^\ell)$, for every $\ell\in\{k+4, \ldots, 2k+5\}\setminus\{2k+4\}$. Besides the edges of the paths, the tree $T(S)$ also has the edges $wv_j$ and $wu_1^\ell$, for every $j\in\{1,\ldots, p\}$ and every $\ell\in\{k+4,\ldots,2k+5\}\setminus\{2k+4\}$.

The function $w: E(T(S))\to \mathbb{N} \setminus \{0\}$ assigns weight $\ell$ to every edge whose both endpoints lie in $P_\ell$ and weight $1$ to the edges $wu_1^\ell$, for every $\ell\in\{k+4,\ldots,2k+5\}\setminus\{2k+4\}$. The weight of the edges $wv_j$ is $i_j$, for every $j\in\{1,\ldots, p\}$. Finally, the edges $ww_1$, $w_1w_2$ and $w_2w_3$ have weight $k+5$, $2k+6$ and $2k+6$, respectively. A representation of $T(S)$ is depicted in Figure~\ref{fig:np-c-trees}.
Let us prove that $(S,k)$ is a {\sf yes}-instance of {\sc Subset Sum} if, and only if, $\po(T(S),w)\leq 2k+6$.

 \begin{figure}[h!]
    \centering
    \includegraphics[scale=0.18]{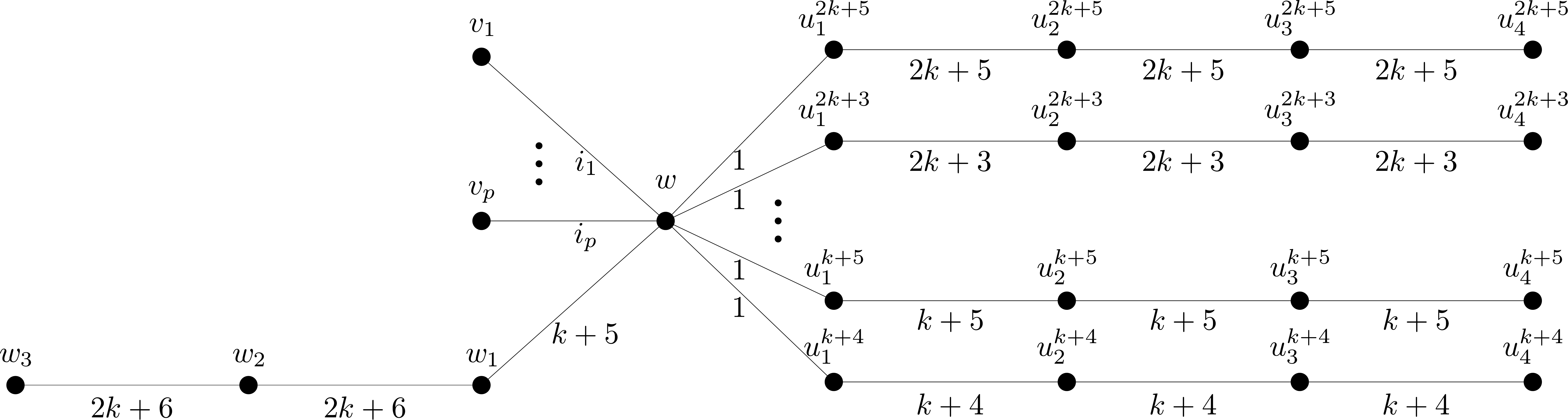}
    \caption{A representation of a tree $T(S)$.}
    \label{fig:np-c-trees}
  \end{figure}

Suppose first that $(S,k)$ is a {\sf yes}-instance and let $S'\subseteq S$ be such that $\sum_{i_j\in S'} i_j= k$. Let us build a proper $(2k+6)$-orientation of $(T(S),w)$. Orient the edges of the paths $P_\ell$ and the edges $wu_1^\ell$ in a way that vertices $u_3^\ell$ have indegree zero and the vertices $u_1^\ell$ have inweight $\ell+1$, for every $\ell\in\{k+4,\ldots,2k+5\}\setminus\{2k+4\}$. Orient the edge $ww_1$ toward $w$ and the other edges of $P^*$ so that the inweight of $w_2$ is zero. Finally, orient the edge $wv_j$ toward $w$ if, and only if, $i_j\in S'$, for every $j\in\{1,\ldots,p\}$. Let us check that this is a proper $(2k+6)$-orientation. Since all weights are positive, the vertices of inweight zero ($w_2$ and $u_3^\ell$, for every $\ell\in\{k+4,\ldots,2k+5\}\setminus\{2k+4\}$) cannot have a neighbor with its inweight. The vertices $u_1^\ell$ have one unit of inweight more than $u_2^\ell$ and, since no vertex $u_1^\ell$ has inweight equal to $2k+5$, none of them has the same inweight as $w$. In fact, $w$ has inweight $2k+5$ due to the edge $w_1w$ and the edges $wv_j$, for every $j\in\{1,\ldots, p\}$. Finally, since $i_j<k$, no neighbor of $w$ has its inweight.

Conversely, suppose that $(T(S),w)$ has a weighted $(2k+6)$-proper orientation $D$. Define $S' = \{i_j\in S\mid v_jw\text{ is oriented toward }w\}$. We claim that $\sum_{i_j\in S'} i_j = k$ and thus $(S,k)$ is a {\sf yes}-instance of  {\sc Subset Sum}. Since the  inweight of every vertex of $T(S)$ in $D$ is upper-bounded by $2k+6$, note that the edges of $P_\ell$ and the edges $wu_1^\ell$ must necessarily be oriented so that all vertices $u_3^\ell$ have inweight  zero and the vertices $u_1^\ell$ have inweight $\ell+1$, for every $\ell\in\{k+4,\ldots,2k+5\}\setminus\{2k+4\}$. Consequently, $w$ has neighbors with inweight  $k+5, \ldots, 2k+4, 2k+6$. By a similar analysis, one can deduce that the edge $w_1w$ must be oriented toward $w$. Thus, $w$ must have inweight  exactly $2k+5$ in $D$. Therefore the edges $wv_j$ that are oriented toward $w$ add to the inweight of $w$ exactly $k$ units.
\end{proof}

\section{Graphs of bounded treewidth}
\label{sec:bounded-tw}

In this section we focus on graphs of bounded treewidth. We provide the $\xp$ algorithm in Subsection~\ref{sec:DP} and the $\W[1]$-hardness proof in Subsection~\ref{sec:W[1]hard-tw}.

\subsection{Dynamic programming algorithm}
\label{sec:DP}


In this subsection, we provide a dynamic programming algorithm to determine the weighted proper orientation number of an edge-weighted graph $(G,w)$ with treewidth at most $\tw$.

Let  $(T,(X_t)_{t\in V(T)})$ be a nice tree decomposition of $G$. We recall that given a node $t\in V(T)$,  $T_t$ is the subtree of $T$ rooted at $t$, and  $G_t$ is the subgraph of $G$ induced by $\bigcup_{t'\in V(T_t)}X_{t'}$. Let $k$ be  a positive integer. Let us define  $X_t = \{v_1,\ldots,v_p\}$ where $p=|X_t|$, and  let $\gamma = (D', a_1,d_1,\ldots,a_p,d_p)$ be a tuple such that $D'$ is an orientation of $G[X_t]$, and $a_i,d_i$ are non-negative integers with $a_i\le d_i\le k$, for every $i\in \{1,\ldots,p\}$. We say that an orientation $D$ of $G_t$ \emph{agrees with $\gamma$} if the edges in $G[X_t]$ are oriented in the same way in $D$ and $D'$, and $w^-_D(v_i)= a_i$ for every $i\in \{1,\ldots,p\}$. Finally, we say that $D$ \emph{realizes $\gamma$} if $D$ agrees with $\gamma$, and the coloring $f_{D,\gamma}$ defined below is a proper coloring of $G$.

\[f_{D,\gamma}(v) = \left\{\begin{array}{ll}
w^-_D(v)& \mbox{, if $v\in V(G_t)\setminus X_t$, and}\\
d_i & \mbox{, if $v = v_i$}
\end{array}\right.\]

Now, we define the following:

\[Q_t(\gamma) = 1 \mbox{ iff there exists an orientation $D$ of $G_t$ that realizes $\gamma$.}\]

Observe that, if $r$ is the root of $T$, then the following holds:

\begin{proposition}\label{prop:DP}
$(G,w)$ admits an orientation $D$ such that $\mu^-(D) \leq k$  if and only if $Q_r(\gamma) = 1$, for some entry $\gamma$ of the type $(D, d_1,d_1,\ldots, d_{|X_r|}, d_{|X_r|})$, where $d_i\le k$ for every $i\in \{1,\ldots,|X_r|\}$.
\end{proposition}

Now, we present the main result of this section. We assume that a nice tree decomposition of width at most $\tw$ is given along with the input graph. This assumption is safe, since by the algorithm of Bodlaender et al.~\cite{DDFLP16} we can compute a (nice) tree decomposition of width at most $5\tw$ of an $n$-vertex graph of treewidth at most $\tw$ in time $2^{\Ocal(\tw)} \cdot n$, and this running time is asymptotically dominated by the running time given in Theorem~\ref{thm:DP-treewidth}.

\begin{theorem}\label{thm:DP-treewidth}
Given an edge-weighted graph $(G,w)$ together with a nice tree decomposition of $G$ of width at most $\tw$, and a positive integer $k$, it is possible to decide  whether $\po(G,w) \leq k$ in time $\Ocal(2^{\tw^2}\cdot k^{3\tw}\cdot \tw \cdot n)$.
\end{theorem}
\begin{proof}
By Proposition~\ref{prop:DP}, it suffices to prove that we can compute $Q_t(\gamma)$ for every node $t \in V(T)$ in time $\Ocal(k^{\tw})$, where $\gamma$ is defined as before.  Indeed, since there are $\Ocal(2^{\tw^2}\cdot k^{2\tw})$ possible entries $\gamma$ and $\Ocal(\tw \cdot n)$
nodes in $V(T)$~\cite{Kloks94}, the theorem follows.  Let us analyze the possible types of nodes in the given nice tree decomposition of $G$.

Suppose first that $t$ is a leaf. Then, $D'$ is the only allowed orientation of $G[X_t] = G_t$; hence, it suffices to test whether $D'$ agrees with $\gamma$, and whether $f_{D',\gamma}$ is a proper coloring of $G_t$. This takes time $\Ocal(\tw^2)$.

Now, suppose that $t$ is an introduce node. Let $t'$ be the child of $t$ in $T$ and suppose, without loss of generality, that $X_t = X_{t'}\cup \{v_p\}$. First, suppose that $d_p$ is different from $d_i$  for every $v_i\in N(v_p)$, as otherwise $Q_t(\gamma)$ is trivially~0. Let $D''$ be equal to $D'$ restricted to $X_{t'}$ and, for each $i \in \{1,\ldots, p-1\}$, let $a'_i$ be equal to $a_i - w(v_pv_i)$ if $\overrightarrow{v_pv_i}\in D'$, or be equal to $a_i$ otherwise. Observe that there exists an orientation $D$ of $G_t$ that realizes $\gamma$ if and only if $Q_{t'}(D'',a'_1,d_1,\ldots,a'_{p-1}, d_{p-1})$ 
equals~1. Hence, it suffices to verify this entry in $Q_{t'}$.

Now, suppose that $t$ is a forget node, and let $t'$ be the child of $t$ in $T$ and $v\in V(G)$ be such that $X_t = X_{t'}\setminus \{v\}$. Observe that $G_t = G_{t'}$. Thus, if $D$ is an orientation of $G_t$ that realizes $\gamma$, then $D$ is an orientation of $G_{t'}$ that realizes $\gamma' = (D'', a_1,d_1,\ldots, a_p,d_p, d,d)$, where $D''$ equals $D$ restricted to $G[X_{t'}]$, and $d = w^-_D(v)$; in other words, the entry $Q_{t'}(\gamma')$ equals~1. Conversely, if any such entry equals~1, then $Q_t(\gamma)$ also equals~1. Therefore, it suffices to verify all the entries of $Q_{t'}$ whose orientation of $X_{t'}$ is an extension of $D'$ and whose values related to $v$ are equal; there are $\Ocal(2^{\tw}\cdot k)$ such entries.

Finally, suppose that $t$ is a join node, and let $t_1,t_2$ be its children. Let $D$ be an orientation of $G_t$ that realizes $\gamma$, and denote by $D_i$ the orientation $D$ restricted to $G_{t_i}$, for $i=1$ and $i=2$. For each $j\in \{1,\ldots, p\}$, let $o_j$ be the inweight of $v_j$ in  $D$ restricted to $X_t$, and $a^i_j$ be the inweight  of $v_j$ in $D_i$, $i\in \{1,2\}$.  Observe that $a^1_j+a^2_j = a_j+o_j$ for every $j\in \{1,\ldots,p\}$. Finally, for $i\in\{1,2\}$, let $\gamma_i = (D', a^i_1,d_1,\ldots,a^i_p,d_p)$. Observe that $D_i$ agrees with $\gamma_i$ and, since $V(G_{t_1})\cap V(G_{t_2}) = X_t$, the colorings $f_{D_1,\gamma_1}$ and $f_{D_2,\gamma_2}$ are equal to $f_{D,\gamma}$ restricted to $G_{t_1},G_{t_2}$, respectively; hence these are proper colorings, which means that $Q_{t_i}(\gamma_i)$ equals~1, for $i \in \{1,2 \}$. Conversely, if $\gamma_1 = (D',a^1_1,d_1,\ldots,a^1_p,d_p)$ and $\gamma_2 = (D',a^2_1,d_1,\ldots,a^2_p,d_p)$ are such that $Q_{t_1}(\gamma_1) = Q_{t_2}(\gamma_2) = 1$ and $a^1_i+a^2_i = a_i + o_i$ for every $i\in\{1,\ldots, p\}$, then we can conclude that $Q_t(\gamma)$ equals~1. Therefore, since there are $\Ocal(k)$ possible combinations of values $a^1_i,a^2_i$ for each $i\in\{1,\ldots,p\}$, we can compute $Q_t(\gamma)$ in time $\Ocal(k^{\tw})$.
\end{proof}

\subsection{$\W[1]$-hardness parameterized by treewidth}
\label{sec:W[1]hard-tw}

In this subsection, we present a parameterized reduction from the \textsc{Minimum Maximum Indegree} problem to the \textsc{Weighted Proper Orientation} problem, which proves that the latter problem is $\W[1]$-hard when parameterized by the treewidth of the input graph.


If all edge weights are identical,  Asahiro et al.~\cite{AsahiroMO11} showed that \textsc{Minimum Maximum Indegree} can be solved in polynomial time. Szeider~\cite{SzeiderACM11} showed that, on graphs of treewidth $\tw$, the problem can be solved in time bounded by a polynomial whose degree depends on $\tw$, provided that the weights are given in unary. Later, Szeider~\cite{Szeider11} showed that this dependence is necessary, that is, that \textsc{Minimum Maximum Indegree} is $\W[1]$-hard parameterized by the treewidth of the input graph.

\begin{theorem}\label{th:W[1]hard-tw}
The \textsc{Weighted Proper Orientation} problem is $\W[1]$-hard parameterized by the treewidth of the input graph $G$, even if the weights are polynomial in the size of $G$.
\end{theorem}
\begin{proof}
Let $(G,w)$ be a weighted graph and $k$ a positive integer. We present a parameterized reduction from the \textsc{Minimum Maximum Indegree} problem parameterized by the treewidth of $G$, which is $\W[1]$-hard~\cite{Szeider11}. We assume that there is no edge with weight greater than $k$, as otherwise one can safely conclude that we are dealing with a {\sf no}-instance.

By multiplying all the edge weights by two and setting $k' = 2k$, we clearly obtain an equivalent instance where all the edge weights are even. Hence, we assume henceforth that all the edge weights, as well as the integer $k$, are even. We call such instances \emph{even}.  We now prove that we can also assume the following property of the instance:

\begin{itemize}
\item[{\footnotesize$\bigstar$}] For every edge $e= uv \in E(G)$ such that $w(e) = k$, it holds that
  $$
    \sum_{y \in N(u)}{w(uy)} < 2k \text{\ \ \ \ and \ \ \ \ } \sum_{y \in N(v)}{w(vy)} < 2k.$$
\end{itemize}

Indeed, given an even instance $(G,w,k)$ of \textsc{Minimum Maximum Indegree}, we define another even instance $(G',w',k')$. The graph $G'$ is obtained from $G$ by attaching a new triangle $(v,v_1,v_2)$ to every vertex $v$ of $G$. The weights of the edges of $G$ remain unchanged and, for each triangle $(v,v_1,v_2)$, we give weight $2$ to $vv_1$ and $vv_2$, and weight $k+2$ to $v_1v_2$. Finally, we set $k' = k+2$. It is easy to check that $(G,w,k)$ and $(G',w',k')$ are equivalent instances. Indeed, observe that a proper orientation $D$ of $G$ with $\mu^-(D)\le k$ can be completed into an orientation $D'$ of $G'$ with $\mu^-(D')\le k+2$ by orienting $\overrightarrow{v_1v},\overrightarrow{vv_2},\overrightarrow{v_2v_1}$ for every $v\in V(G)$. Conversely, if $D'$ is an orientation of $G'$ with $\mu^-(D')\le k+2$, then for every $v\in V(G)$, at least one between $vv_1,vv_2$ must be oriented toward $v$, which means that $D'$ restricted to $G$ has inweight at most $k$.  Note that $(G',w',k')$ satisfies Property~{\footnotesize$\bigstar$}.

Now, let $(G,w,k)$  be an even instance of \textsc{Minimum Maximum Indegree} satisfying Property~{\footnotesize$\bigstar$}. We construct an instance $(G',w',k)$ of \textsc{Weighted Proper Orientation} as follows. We define $G'$ and $w'$ from $G$ and $w$ by replacing each edge $e = uv$ by the gadgets depicted in Figure~\ref{fig:reduction-tw}. Namely, if $w(e) = w < k$ (resp. $w(e)=k$), we replace it by the gadget and weights shown in Figure~\ref{fig:reduction-tw}(a) (resp. Figure~\ref{fig:reduction-tw}(b)).

\begin{figure}[h!]
    \centering
    \includegraphics[scale=1.1]{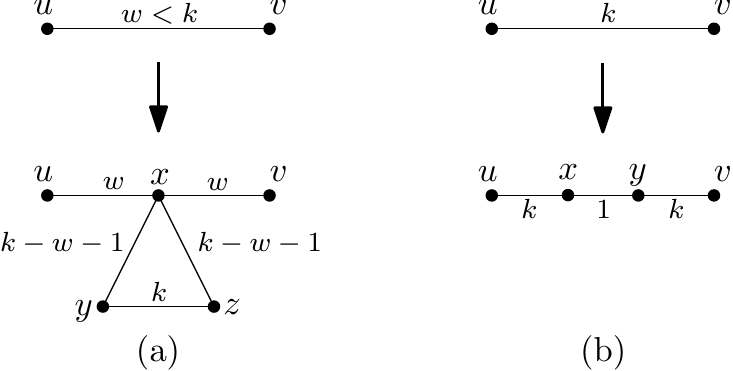}
    \caption{Gadgets in the proof of Theorem~\ref{th:W[1]hard-tw}. The weights are depicted near the edges.}
    \label{fig:reduction-tw}
  \end{figure}

Note that all the gadgets introduced so far do not increase the treewidth of the original instance of \textsc{Minimum Maximum Indegree}, assuming that it is at least two. To conclude the proof, we claim that the instances $(G,w,k)$ and $(G',w',k)$ are equivalent. Note that the only difference between the two problems is the desired orientation needing to be {\sl proper} or not.

Assume first that $(G,w,k)$  is a {\sf yes}-instance of \textsc{Minimum Maximum Indegree}, and let $D$ be the corresponding  orientation of $G$. We define from $D$ a proper orientation $D'$ of $G'$ satisfying $\po(G',w')\leq k$ as follows. Let $e = uv \in E(G)$ be an edge such that $w(e) = w < k$ (cf. Figure~\ref{fig:reduction-tw}(a)), and assume w.l.o.g. that  $\overrightarrow{uv}$ is an arc of $D$. In this case, in $D'$, the edges of the corresponding gadget will be replaced by the arcs  $\overrightarrow{ux}, \overrightarrow{xv}, \overrightarrow{xz}, \overrightarrow{zy}, \overrightarrow{yx}$. On the other hand, let $e = uv \in E(G)$ be an edge such that $w(e) = k$ (cf. Figure~\ref{fig:reduction-tw}(b)), and assume w.l.o.g. that $\overrightarrow{uv}$ is an arc in $D$. Then, in $D'$, the edges of the corresponding gadget will be replaced by the arcs $\overrightarrow{ux}, \overrightarrow{xy}, \overrightarrow{yv}$.

In the first case, we have that $w^-_{D'}(u) = w^-_{D}(u)$, $w^-_{D'}(v)= w^-_{D}(v)$, $w^-_{D'}(x) = k-1$, $w^-_{D'}(y) = k$, and $w^-_{D'}(z) = k - w- 1$. Since $w^-_{D'}(u)$, $w^-_{D'}(v)$, and $w^-_{D'}(y)$ are even, and $w^-_{D'}(x)$ and $w^-_{D'}(z)$ are odd, we have that $w^-_{D'}(u) \neq w^-_{D'}(x)$, $w^-_{D'}(x) \neq w^-_{D'}(v)$, $w^-_{D'}(x) \neq w^-_{D'}(y)$, and $w^-_{D'}(y) \neq w^-_{D'}(z)$. Also, as $2 \leq w < k$ and $k$ is even, $0 < w^-_{D'}(z) < w^-_{D'}(x)$.

In the second case, we have that $w^-_{D'}(u) = w^-_{D}(u)$, $w^-_{D'}(v)= w^-_{D}(v) = k$, $w^-_{D'}(x) = k$, and $w^-_{D'}(y) = 1$. As $k$ is even, clearly $w^-_{D'}(x) \neq w^-_{D'}(y)$ and $w^-_{D'}(y) \neq w^-_{D'}(v)$. Also, since $(G,w,k)$ satisfies Property~{\footnotesize$\bigstar$}, necessarily $w^-_{D}(u) < k$, and therefore $w^-_{D'}(u) < w^-_{D'}(x)$.

\medskip

Conversely, assume that $(G',w',k)$  is a {\sf yes}-instance of \textsc{Weighted Proper Orientation}, and let $D'$ be the corresponding  orientation of $G'$. We define from $D'$ an orientation $D$ of $G'$ satisfying that $w^-_D(v) \leq k$ for every $v \in V(G)$.

Recall that every possible weight $w$ is even, as well as $k$. In the first case (cf. Figure~\ref{fig:reduction-tw}(a)), note that the $\overrightarrow{xy}$ and $\overrightarrow{xz}$ cannot be simultaneously arcs of  $D'$, as in that case one of $y$ and $z$, say $y$, would satisfy $w^-_{D'}(y) = 2k - w - 1 > k$, a contradiction. Hence, assume w.l.o.g. that   $\overrightarrow{yx}$ is an arc, which implies that   $\overrightarrow{ux}$ and $\overrightarrow{vx}$ cannot be both arcs of $D'$, as in that case $w^-_{D'}(x) \geq k + w - 1 > k$, a contradiction.
Therefore, at least one of $\overrightarrow{xu}$ and $\overrightarrow{xv}$ is an arc of $D'$. If exactly one of them is, say $\overrightarrow{xu}$, we replace the edge $uv$ by the arc $\overrightarrow{vu}$ in $D$. Otherwise, if both $\overrightarrow{xu}$ and $\overrightarrow{xv}$ are arcs of $D'$, in $D$, we replace the edge $uv$ by one of the possible arcs arbitrarily.

In the second case (cf. Figure~\ref{fig:reduction-tw}(b)), note that if  $\overrightarrow{xy}$ is an arc of  $D'$, then $\overrightarrow{yv}$  is necessarily an arc of $D'$, as otherwise $w^-_{D'}(y) = k +1$, a contradiction. Based on this remark, if the $\overrightarrow{xy}$ (resp. $\overrightarrow{yx}$) is an arc of  $D'$, we replace the edge  $uv$ by the arc $\overrightarrow{uv}$ in $D$ (resp. $\overrightarrow{vu}$).

In both cases, it holds that $w^-_D(u) \leq w^-_{D'}(u) \leq k$ and $w^-_D(v) \leq w^-_{D'}(v) \leq k$, as required.

\medskip

Finally, by the $\W[1]$-hardness reduction of Szeider~\cite{Szeider11} for \textsc{Minimum Maximum Indegree}, we may assume that the value of $k$ in the original instance $(G,w,k)$ is bounded
by a polynomial on the size of $G$. Therefore the proof above indeed rules out the existence of an algorithm for
\textsc{Weighted Proper Orientation}
 running in time $f(\tw) \cdot (k \cdot n)^{\Ocal(1)}$ for any computable function $f$, provided that $\fpt \neq \W[1]$.
\end{proof}

\section{Conclusions and further research}
\label{sec:further}

%
%

In this article, we introduced the parameter weighted proper orientation number of an edge-weighted graph, and we studied its computational complexity on trees and, more generally, graphs of bounded treewidth. In particular, we proved that the problem is in $\xp$ and $\W[1]$-hard parameterized by the treewidth of the input graph. While the $\xp$ algorithm can clearly be applied to the unweighted version as well, it is still open whether determining the proper orientation number is \fpt parameterized by treewidth. It is worth mentioning that our positive results still apply if the edge-weights are positive {\sl real} numbers.

Another avenue for further research is to generalize the upper bounds given by Ara\'ujo et al.~\cite{AHLS16} on cacti to the weighted version, in the same way as Lemma~\ref{lem:boundedby4K} generalizes the bounds given by Ara\'ujo et al.~\cite{ACR+15} and Knox et al.~\cite{KMM+17} on trees. More generally, it would be interesting to know whether there exists a function $f: \mathds{N} \times \mathds{N} \to \mathds{N}$ such that, if we denote by $w_{\max}$ the maximum edge-weight of a weight function and by $\tw$ the treewidth of the input graph, for any edge-weighted graph $(G,w)$ it holds that $\po(G,w) \leq f(\tw, w_{\max})$. Lemma~\ref{lem:boundedby4K} shows that, if such a function $f$ exists, then $f(1, w_{\max}) \leq 4 \cdot w_{\max}$. Note that the existence of $g(\tw) := f(\tw, 1)$ was left as an open problem in~\cite{ACR+15}.

Finally, we refer the reader to~\cite{ACR+15} for a list of open problems concerning the proper orientation number, noting that most of them also apply to the weighted proper orientation number. In particular, it is not known whether there exists a constant $k$ such that $\po(G) \leq k $ for every planar graph $G$. Another possibility is to try to generalize the (few) positive results for the proper orientation number to the weighted version, such as the case of regular bipartite graphs~\cite{AD13}, or to prove stronger hardness results.

\bibliographystyle{abbrv}
\bibliography{properorientation-bib}
\end{document}